\documentclass[sigconf, nonacm]{acmart}
\AtBeginDocument{%
  }

\usepackage{listings}
\definecolor{codegray}{rgb}{0.5,0.5,0.5} 

\usepackage{url}
\usepackage{csquotes} 
\usepackage{wasysym} 

\newtheorem{theorem}{Theorem}
\newtheorem{lemma}{Lemma} 
\newtheorem{requirement}{Requirement} 

\begin{document}

\title{AutoTam: Specifying Secure Protocol Implementations with Tamarin Model Generation}


\author{Johannes Wilson}
\affiliation{%
  \institution{Sectra Communications}
  \city{Linköping}
  \country{Sweden}
}
\affiliation{%
  \institution{Linköping University}
  \city{Linköping}
  \country{Sweden}
}
\email{johannes.wilson@liu.se}

\author{Mikael Asplund}
\affiliation{%
\institution{Linköping University}
 \city{Linköping}
 \country{Sweden}}
\email{mikael.asplund@liu.se}

\author{Niklas Johansson}
\affiliation{
  \institution{Sectra Communications}
  \city{Linköping}
  \country{Sweden}
}
\affiliation{%
  \institution{Linköping University}
  \city{Linköping}
  \country{Sweden}
}
\email{niklas.johansson@sectra.com}


\begin{abstract}
Formal verification is a challenging but important task for ensuring the security of cryptographic protocols. While modern protocol verification tools significantly reduce verification effort, modelling remains challenging to practitioners without a background in formal verification. In addition, transferring verification results to a concrete protocol implementation requires expert knowledge.

In this paper, we present a novel language-first method for verification of trace properties using a domain-specific language for protocol implementations. We target the Tamarin prover for verification, and we prove that verified universal trace properties translate back to the implementation. We additionally integrate symbolic execution in order to analyse the memory safety of protocol implementations.
We use our tool to implement and generate accurate models for a signed Diffie-Hellman protocol, and for the WireGuard VPN protocol. Our WireGuard implementation is interoperable with existing implementations when using our interpreter, and achieves acceptable performance. We formally prove our implementations secure using a combination of symbolic execution and verification of the generated Tamarin models.
\end{abstract}



\maketitle

\section{Introduction}
Cryptographic protocols are vital to the security of network communication.
The design of secure protocols has proven difficult, as demonstrated by numerous vulnerabilities discovered in both proposed and deployed protocols, such as in the original Needham-Schroeder protocol~\cite{Lowe1996NeedhamSchroeder}, prior versions of SSL/TLS~\cite{Bhargavan2014triplehandshake}, 
the EMV standard~\cite{basin2021emvtamarin}, and many others. To address these challenges, symbolic protocol verification based on a Dolev-Yao model~\cite{DolevYao} was developed 
as a method to formally verify security properties of protocols, proving security against strong attackers for a wide class of attacks. 
These methods can enjoy a great deal of automation, thanks to the development of several automated verification tools such as Tamarin~\cite{TamarinProver} and ProVerif~\cite{ProVerif}.
Today, the security of new proposed protocols is often argued using verifiable symbolic models, as was done with the WireGuard VPN protocol~\cite{donenfeld2017ndsswireguard} proposed in 2017, for which a security proof~\cite{donenfeld2017formalwireguard} was later made using the Tamarin prover.

The Dolev-Yao model, while practical for verification, relies on a highly abstract protocol model. Much effort has been spent on designing methods to bridge the gap between symbolic models and concrete protocol implementations, providing verification of high-level protocol security properties also for implementations. 
Despite the availability of advanced analysis tools and methods applicable to a wide variety of protocol implementations, industrial application remains low. 
This in part because the required verification expertise is significant. To improve adoption in industry, methods must be streamlined further in order to minimise the needed expertise. 

To streamline verification and minimize effort, while retaining a well-defined relationship with the formal model, a new approach is needed.
We propose a novel language-first approach, where the symbolic Dolev-Yao model and implementation language are tied together through the use of a domain-specific language. 
We motivate this in three ways. Firstly, a language designed to align with the abstraction level of high-level security properties ensures that there is a clear connection to the verified properties, creating formal guarantees which are more easily understood. Secondly, an appropriately designed language ensures that entire implementations can be symbolically modelled, including implementation state machines, as opposed to only verifying core segments like message parsing or formatting. This limits the need for expertise during the implementation phase to align the code with verification. Finally, a protocol-domain-specific language can naturally encode the symbolic model, allowing automatic and sound translation to a model, only requiring formal method expertise at the verification phase itself.  

The most similar existing approaches perform model extraction from a general-purpose language~\cite{bhargavan2008fsquaretoproverif, bhargavan2012defensivejs, Almousa2015SPS, Kobeissi2017ProScript, bhargavan2024hax, bhargavan2025protocolrust}. Typically, limits on loops and recursion must be imposed to retain a tractable model, and as a consequence most approaches leave some parts of the protocol state machine to manual modelling assumptions. We argue that for an implementation, it is desirable to additionally verify the state machine, considering that historical protocol vulnerabilities such as SKIP~\cite{Beurdouche2017statemachineTLS} were caused by incorrectly implemented state machines. ProScript explicitly disallows all forms of recursion or loops~\cite{Kobeissi2017ProScript}, while for SPS, the input language implicitly requires a linear protocol~\cite{Almousa2015SPS}. This limits the use of the tools to only verification of linear handshakes, and not repeatable message steps such as transport phases, nor more complex stateful protocols. While in theory fs2pv~\cite{bhargavan2008fsquaretoproverif} does not forbid recursion, recursion is discouraged for verification reasons and none of its case studies apply the tool to a non-linear protocol. Defensive JavaScript~\cite{bhargavan2012defensivejs} allows loops but faces other challenges in terms of generated model size~\cite{Kobeissi2017ProScript}. The verification framework hax~\cite{bhargavan2024hax, bhargavan2025protocolrust} provides automatic translation from annotated Rust functions to ProVerif. In this framework, a scenario, i.e. the initialisation of both roles, must be manually added to the resulting model.

In this paper, we present AutoTam, which consists of a domain-specific language in which protocol implementations can be written, symbolic models for verification can be extracted, and automated testing using symbolic execution can be performed. Using our tool, much of the task of producing a secure protocol implementation can be concentrated to writing the protocol implementation in AutoTam. The tool is not limited to linear protocols. In fact, we are able to encode arbitrary finite state machines to describe the message order. This guarantees we can also specify the unbounded transport phases of protocols. 

Compared to existing work, our approach is novel. Protocols can largely be defined by their protocol state-machine specifying what messages to accept in each state and how to transition between states, the parsing and message serialisation for each message, and the cryptographic operations to be performed in each state. We design our high-level domain-specific language with these components in mind, ensuring we can concisely encode each part in a way that follows the imperative style code that developers are used to. The resulting AutoTam protocol descriptions concisely describe implementation behaviour with enough detail to serve as reference implementations. We implement an interpreter for executing the protocol descriptions directly to allow for interoperability testing.
For verification, we implement a process for translating a set of AutoTam protocol descriptions into a Tamarin model. We have spent considerable effort into designing a translation that is both sound and produces readable and verifiable models.

We believe our tool will be useful during the development of new protocols.
Once AutoTam protocol role descriptions are written, symbolic verification of security properties for confidentiality and authenticity comes largely for free, as the tool produces verification-ready Tamarin models. Our soundness property relates such properties back to the protocol descriptions.
Furthermore, our tool can provide templates for common security properties, providing a way to automatically add them to the generated Tamarin model. This makes it easy for users to verify common confidentiality and authenticity properties.

Our language is designed with the soundness of translation in mind. Our soundness property states that a trace property proven in the generated model translates to a proof that the property holds also during execution. We provide a proof of the soundness of this translation under standard assumptions on cryptography according to the Dolev-Yao~\cite{DolevYao} symbolic protocol model. Appropriately selected security properties can rule out a large number of implementation errors, including state-machine bugs.

Finally, we fully integrate with a dynamic symbolic execution engine in a way that allows automated analysis of implementations written in AutoTam. In our case studies, this analysis produces complete coverage tests of the implementations when receiving arbitrary network input. This way, we do not have to assume the safe behaviour of our language interpreter or protocol implementations, since they can both be analysed using symbolic execution for memory safety.

\subsection{Contributions}

We can summarize the contributions of our work as follows:

\begin{itemize}
    \item We describe the design of the AutoTam language, a domain-specific language for cryptographic protocols which is fully executable by our interpreter.
    \item We present our procedure for extracting Tamarin models from the executable implementations, and provide a soundness proof of the translation.
    \item We present two case studies on using our language and tool to implement, extract models for, and prove the security of a signed Diffie-Hellman key exchange protocol and the WireGuard VPN protocol.
    \item We provide a fully integrated symbolic library for performing analysis of implementations using symbolic execution, allowing us to rule out most types of memory errors in our case study protocols.
\end{itemize}

\section{Background} \label{sec:autotam_background}

In this section we give a short introduction to the ideas of symbolic protocol verification and the Tamarin prover, a cryptographic protocol verifier in the symbolic model which uses heuristics to find proofs or counterexamples to security properties.

\subsection{Symbolic Protocol Verification} \label{subsec:sym_protocol_verif}

Symbolic protocol verification follows work originally by Dolev and Yao~\cite{DolevYao} where cryptographic functions are represented symbolically and the adversary is modelled as an active network adversary. In contrast to a computational model, a symbolic model does not model concrete bitstrings. Instead, messages, variables, role names and keys are all represented as abstract terms, making verification much more amenable for automation. 

Each protocol role will describe a series of steps, where each step may interact with the global protocol trace. In the Dolev-Yao model~\cite{DolevYao}, the adversary is an active network adversary and may intercept, remove, duplicate, and construct its own messages on the network. Security properties can be formulated as trace properties on the global network trace. For example, secrecy properties can be expressed as claims that no possible trace exists where the adversary will be able to directly derive a message term which was at some point believed to be secret. A proof of such a property must then show that no matter the order in which the protocol roles perform the protocol, and no matter how the adversary modifies or constructs messages, the property always holds. Such verification problems are known to be undecidable in the general case when protocol participants can generate new unique values~\cite{mitchell1999undecidability, durgin2004multiset}. Therefore, automated provers must rely on heuristics or perform bounded verification.

\subsection{Tamarin}

Tamarin is a state-of-the art symbolic protocol verification tool which has been used in formal verification of several important security protocols, such as TLS~\cite{cremers2016tls13, cremers2017tls13}, the EMV standard~\cite{basin2021emvtamarin}, and 5G Key-Agreement~\cite{cremers2019formal5GAKA}. Tamarin can perform unbounded verification, meaning properties can be verified for an arbitrary number of protocol runs.
Tamarin provides a very general way to model stateful protocols in the form of multiset rewrite rules, where each rewrite rule may affect the global trace and adversary knowledge~\cite{meier2013advancing}.

We write a rewrite rule on the form [\textit{l}]\texttt{-}[\textit{a}]\texttt{->}[\textit{r}], where multisets \textit{l} and \textit{r}, containing terms known as \textit{facts}, represents the pre- and post-condition of the rule, while the multiset \textit{a} contains terms known as \textit{actions}, describing events visible on the global protocol trace. Application of a rule requires that there exists some instantiation of the variables in \textit{l} which matches available state facts of the global protocol state. 

Cryptographic operations are modelled according to an equational theory, for example describing the relationship between encrypted and decrypted terms. 
The network channel is modelled using special facts \texttt{In} and \texttt{Out} whose contents describe the input and output from the network respectively. The adversary may perform any publicly known operations such as encryption and decryption as described by the equational theory of the model. Tamarin supports types to distinguish public variables, prefixed with a \$ character or within single quotes in case of a constant, and fresh variables, prefixed with a $\sim$ character. Public variables are known to the adversary, while fresh variables may only be generated by the rule \texttt{Fr}. \texttt{Fr} allows modelling uniquely generated numbers and private keys. We show an example Tamarin rule below.

\begin{verbatim}
rule RuleName:
    let c = <a, ~b> in
    [ In(a), Fr(~b) ]
    --[ ActionName(a, ~b) ]->
    [ StateNext(c), Out(c) ]
\end{verbatim}

The rule accepts some term \texttt{a} from the network and generates some fresh term $\sim$\texttt{b}. A term \texttt{c} is defined as the concatenation of \texttt{a} and \texttt{b}. The term \texttt{c} is stored in a produced state fact \texttt{StateNext}, and also output to the network. The rule also produces an action fact \texttt{ActionName}. An action fact is not part of the state but is rather used to track the protocol trace. The action fact \texttt{ActionName} associates an application of the rule \texttt{RuleName} with its terms \texttt{a} and \texttt{b}.

In Tamarin, security properties are specified as guarded first order logic formulas on action facts and timepoints. This property specification language allows specifying very general authentication and secrecy properties. Such formulas are called lemmas within Tamarin.  
A special type of lemma known as a restriction can enforce protocol invariants, which may be used in order to implement equality and inequality of terms appearing in rules. We will use two special action facts, \texttt{Eq} and \texttt{NEq}, both with arity two, in order to signify that two terms appearing in a rule must be equal or distinct, respectively.

\section{AutoTam Overview} \label{sec:autotam_overview}

AutoTam is a tool for both interpreting our own domain specific language, the AutoTam language, and for generating Tamarin models from the same language.
We now proceed to present the features of the AutoTam domain specific language and the AutoTam tool. An overview of the AutoTam architecture and its workflow is shown in \figurename~\ref{fig:autoam_overview}. AutoTam operates in three modes: interpreter mode, translation mode, and symbolic execution mode. In interpreter mode, an AutoTam protocol role description written in the AutoTam language is parsed and executed, while in translation mode AutoTam protocol role descriptions are translated to a Tamarin model. In symbolic execution mode, we use the KLEE dynamic symbolic execution engine~\cite{cadar2008kleecoverage} to generate high-coverage test cases for the protocol implementation when executed by our interpreter.

\begin{figure}
    \centering
    \includegraphics[width=0.9\columnwidth]{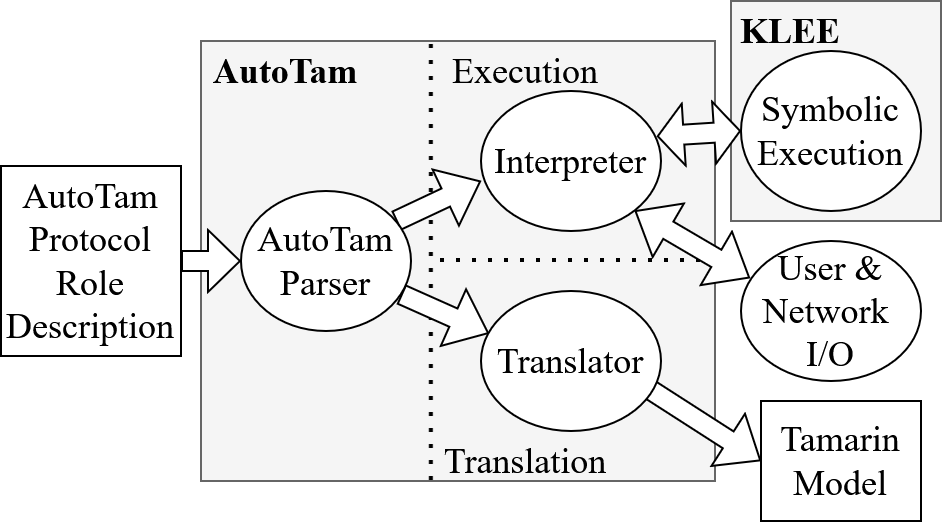}
    \caption{AutoTam Architecture and Workflow.}
    \Description{A series of boxes and ellipses illustrating the workflow of the AutoTam tool. At the far left, a box is labelled AutoTam protocol role description. It leads to an ellipse labelled AutoTam Parser, which in turn leads to both ellipses Interpreter and Translator. AutoTam Parser, Interpreter, and Translator are all contained in a larger box labelled AutoTam. Interpreter connects with two-way arrows to both ellipses Symbolic Execution and User \& Network I/O. Symbolic Execution is contained within a box labelled KLEE. Finally, Translator leads to a box labelled Tamarin Model.}
    \label{fig:autoam_overview}
\end{figure}

An AutoTam protocol role description describes the control flow, cryptographic operations, and network operations performed by a role in the protocol. The protocol role description files are written in the AutoTam language. We present the syntax and execution semantics of the AutoTam language in Section~\ref{sec:language_definition}.
In interpreter mode, AutoTam will accept a single protocol role description file, describing one of the protocol roles. During execution, AutoTam will interface with the user through a message passing mechanism, and with the network through calls to system libraries for reading and writing socket file descriptors.

In translation mode, described in detail in Section~\ref{sec:translation}, AutoTam may instead accept several protocol role description files that together describe all roles in the protocol. In addition, Tamarin lemmas can be built using pre-defined templates for secrecy and authenticity properties. The protocol role description files will then be translated into a single Tamarin file ready for verification.

Symbolic execution mode, described in Section~\ref{sec:symbolic_execution_method}, is achieved by compiling the AutoTam interpreter into LLVM bytecode for analysis using the dynamic symbolic execution engine KLEE~\cite{cadar2008kleecoverage}. The execution traces of the interpreter executing an AutoTam protocol role description file can then be systematically explored under simplifications of the cryptographic operations using KLEE.

\section{AutoTam Language} \label{sec:language_definition}

In this section we present the syntax and semantics of the AutoTam language. The AutoTam language is protocol domain specific.
The language is designed to be both simple and intuitive, yet expressive enough to implement modern cryptographic protocols.
We design our language to support arbitrary protocol state machines while maintaining a natural translation into a symbolic model.

\subsection{Language Control Flow} \label{subsec:lang_control_flow}

Our language design is focused specifically to provide a natural way to describe the state machines of protocols, since complexity in protocol state machines can be a cause of implementation errors~\cite{Beurdouche2017statemachineTLS}. 

To ensure we can soundly extract a model from our implementation, we limit control-flow to a high-level protocol state machine only, avoiding undecidable reachability problems when extracting the possible input-output relations of messages.
This does not in general limit the number of protocols which are possible to implement using our language, since any intra-procedural choice can be implemented as a choice that is part of the protocol state machine.

An AutoTam protocol role description consists of \textit{states} and \textit{transitions}. These states and transitions in turn contain sequences of statements. The transitions encode the state machine of a protocol.
A state may connect with multiple outgoing and incoming transitions. We allow the resulting protocol state machine to form cycles, as this is necessary for implementing protocol transport phases.
The statements of the states and transitions describe the operations to perform during execution. The states are further divided into two parts, an output part and an input part. The general syntax is shown in Listing~\ref{lst:syntax}. A formal description of the syntax is given in Listing~\ref{lst:syntax_backusnaur} in the Appendix.

\begin{lstlisting}[numbers=none, numbersep=4pt, frame=lines, language={C}, basicstyle=\ttfamily\small, caption={AutoTam Syntax}, label={lst:syntax}, commentstyle=\color{codegray}, captionpos=b]
[StateName] {
    Output{
        [Statements]
    }
    Input {
        [Statements]
    }
}
Transition ([FromState], [ToState]) {
    [Statements]
}
\end{lstlisting}

The program may only transmit messages over the network in the output part of a state and may only receive input in the input part. States must be structured so that the output part lies before the input part. Such restrictions on the placement of network operations ensure that each transition in the program code can be safely translated into a single rule which represents a given protocol step. The procedure for translation will be described in Section~\ref{sec:translation}.

We illustrate how the states and transitions build up the program control flow in \figurename~\ref{fig:autotam_states}. In the figure, the output part and input parts of states are represented by the blue left and orange right semicircles respectively. The transitions are represented by the white arrows. The gray boxes which cover both parts of states and transitions represent a list of statements which should be executed for each part of the program. Such lists may be empty.

A failure during the execution of a transition, for example a failure to parse a message according to a message format, will end the execution of the transition, returning the execution of the program to before the transition. The interpreter will then attempt to execute any alternative transitions which are available from the state, following in the order that the transitions are defined. This is the mechanism by which it is possible to specify the conditions under which a transition should or should not be taken. 

Typically, each transition would start by parsing a message received in the preceding state. If multiple messages with different formats may be received at some point in the protocol, then each message would have its own transition, and during execution each of them would be tried. We may view the transitions as providing the branching control flow of our language, but at the abstraction level of the protocol. On the other hand, if an error is encountered during the execution of a state, it is unrecoverable and execution will terminate. 

\begin{figure}
    \centering
    \includegraphics[width=0.9\columnwidth]{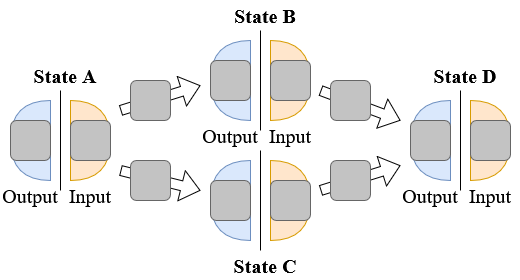}
    \caption{Illustration of AutoTam States and Transitions}
    \Description{A diagram with four states, A, B, C, and D, each drawn as two semicircles, with a thin black line in-between them. The left semicircle is labelled Output while the right is labelled Input. Each semicircles contains a grey rectangle. The states are connected with arrows, going left to right: one from A to B, A to C, B to D, and C to D. Each arrow also has a grey rectangle on top of it.}
    \label{fig:autotam_states}
\end{figure}

We exemplify the traces from \figurename~\ref{fig:autotam_states}. With state names written as $A_O$ for output part of state A, and $T_{AB}$ for the transition from A to B, an example execution trace could proceed as follows: $[A_O,A_I,T_{AB},B_O,B_I,T_{BD},D_O,D_I]$. Alternatively, if an error occurs during execution of transition $T_{AB}$, execution may instead proceed as $[A_O,A_I,T_{AB},T_{AC},C_O,C_I,T_{CD},D_O,D_I]$, where $T_{AB}$ is not fully executed.

To retain soundness during the translation to a Tamarin model which we will present in Section~\ref{sec:translation}, we must impose the following requirements on valid AutoTam programs. These restrictions ensure that a model can fully capture the behaviours of the executable program and are necessary for our soundness proof in Section~\ref{subsec:soundness_translation}.

\begin{requirement} \label{req:no_live_tran}
    A transition cannot assign values to variables which are live (in the usual meaning of variable liveness, see e.g.~\cite{alfred2007compilers}) at the beginning of another transition starting from the same state, as they may get overwritten on failure.
\end{requirement}

\begin{requirement} \label{req:no_stop_out}
    The output part of a state cannot contain statements which may encounter non-network errors, including equality tests, AEAD decryption MAC validation, and signature verification.
\end{requirement}

\subsection{Variables and Functions} \label{subsec:language_variables}

Rather than providing a set of primitive types and operations on these types, our language operates directly on the abstraction level of cryptographic operations and network operations. The only typing in the language exists to track the sizes of variables, so that only variables of the correct size are used in functions as part of ensuring memory safety at runtime. 

All statements in the language consist of a call to some predefined function. 
A function may assign multiple variables and take some other variables as parameters. All statements have the form 
\texttt{v1,...,vn = FunctionName(p1,...,pm)}
where variables \texttt{v1} - \texttt{vn} are assigned by the function, while \texttt{p1} - \texttt{pm} are function parameters.

Available functions in the language describe either operations on network I/O, cryptographic operations, message operations such as message concatenation or division, or user I/O operations. 
The functions themselves are implemented as C functions. During execution the interpreter will first make sure that the sizes of parameters and assigned variables are correct at runtime before calling the appropriate function implementation from a lookup table. Listing~\ref{lst:autotam} shows an example code snippet with two states and a transition between them. 

\begin{lstlisting}[numbers=none, numbersep=4pt, frame=lines, language={C}, basicstyle=\ttfamily\small, caption={Running Example AutoTam Language Protocol Role Description}, label={lst:autotam}, commentstyle=\color{codegray}, captionpos=b]
State1{
    Input{ 
        a = recv(fd,10) 
    }
} Transition (State1, State2) {
    b, c = split(a, 5)
    RequireSize(c, 5)
} State2{
    Output{
        d = GenSecret(5)
        e = Concat(c,d)
        send(fd,e)
    }
    Input{
        f = recv(fd,20)
    }
}
\end{lstlisting}

We will use the code in Listing~\ref{lst:autotam} as a running example throughout the rest of the paper. In Listing~\ref{lst:autotam}, executing \texttt{State1} would call the \textquote{recv} function, reading \texttt{fd} as a socket file descriptor, for at most ten bytes. If completed without error, execution of \texttt{State1} would finish, and the interpreter would look for transitions starting from \texttt{State1}. Assuming the transition \texttt{State1} to \texttt{State2} is the only transition, execution would continue in the transition by calling the split function, followed by a length check. If the variables parsed both had the appropriate length of five bytes, the transition would have successfully completed, transferring execution to \texttt{State2}. In \texttt{State2}, a new random number of five bytes would be generated, which would be concatenated with the five bytes of \texttt{c} parsed earlier, which would then be sent over the same socket. Finally, execution of \texttt{State2} concludes with waiting for the next 20 bytes from the socket.

For the implementation of the cryptographic functions, we used the Cryspen HACL Packages\footnote{\url{https://github.com/cryspen/hacl-packages}} which builds upon the HACL~\cite{zinzindohoue2017hacl, polubelova2020haclxn} formally verified cryptographic library.
The cryptographic functions which we use in this paper are encryption and decryption using the Chacha20 Poly1305 AEAD algorithm used by WireGuard, EC25519 Diffie-Hellman key derivation, Ed25519 signatures, and Blake2s hashes.
Additional pre-defined functions which we provide include message operations such as concatenation and splitting, network I/O, and equality comparisons. For a list of all available pre-defined functions, we refer to the left-hand side of Table~\ref{tab:dsl_translation} in the Appendix.

\section{Translation to Tamarin} \label{sec:translation}

In Translation mode, AutoTam will take several protocol implementation files written in the AutoTam language and generate a single complete Tamarin model which models the behaviour of all involved protocol roles and their interactions with the network. 

In this section, we present the translation process from the AutoTam language to Tamarin models, first describing the translation of control flow in Section~\ref{subsec:translation_rules}, and then presenting the translation of individual statements in Section~\ref{subsec:translation_statements}. In Section~\ref{subsec:soundness_translation} we present our soundness theorem and its proof.

\subsection{Translation to Rules} \label{subsec:translation_rules}

The execution semantics of the input language described in Section~\ref{sec:language_definition} can be directly translated into the trace-based semantics of Tamarin. For each defined transition in the input language, our translator will output a single Tamarin rule. Translation to a rule is performed for an \textquote{extended transition}, which also includes the input part of the previous state and the output part of the succeeding state.

In the running example in Listing~\ref{lst:autotam}, a transition starts from \texttt{State1}, leading to \texttt{State2}. 
The translation of the transition from \texttt{State1} to \texttt{State2} will begin at the input part of \texttt{State1}, include all its statements, continue through all the statements in the transition, and finally conclude with the statements in the output part of \texttt{State2}. The resulting Tamarin rule is presented in Listing~\ref{lst:tamout}. 

\begin{lstlisting}[numbers=none, numbersep=4pt, frame=lines, language={C}, basicstyle=\ttfamily\small, caption={AutoTam Running Example Tamarin Output}, label={lst:tamout}, commentstyle=\color{codegray}, captionpos=b]
rule Role1_State1_to_State2:
    let a = <b, c>
        e = <c, ~d>
    in
    [ State1(fd), In(a), Fr(~d) ]
  --[ A_Role1_State1_State2(a,b,c,~d,e)]->
    [ State2(fd,b,c,~d,e), Out(e) ]
\end{lstlisting}

In addition to the translated individual statements, the produced rule includes state facts in the precondition and postcondition, and an action fact to add to the program trace, as the action facts are necessary for writing lemmas. For our running example, the generated precondition contains the state fact for \texttt{State1}, while the postcondition contains the state fact for \texttt{State2} in order to maintain the execution trace semantics of the implementation.

The Tamarin state facts represent the point of execution in the implementation right in the middle of each state, between execution of the output and input parts of each state. In \figurename~\ref{fig:autotam_states}, these points are represented by the black vertical lines in the middle of states. In the example rule, the state fact \texttt{State2} represents the point of execution after executing the output part statements of \texttt{State2} in the implementation, but before executing any of the statements in the input part of \texttt{State2} in Listing~\ref{lst:autotam}.

As noted by Cheval et al.~\cite{cheval2022sapic}, the number of Tamarin rules in a model can have a significant impact on verification time, and simply joining together compatible rules can reduce verification time considerably.
In a sense, our tool naturally allows the optimisation method described in the extended version~\cite{cheval2022sapicextended}. Instead of producing a rule for each step of execution and then merging compatible rules as far as possible while maintaining the trace semantics, our language makes it natural to start with a large execution step and only subdivide it where strictly necessary.
Our own experience with generating a model for the WireGuard protocol seems to indicate that also the complexity of individual rules can have significant impact on verifiability, suggesting that sometimes it is favourable to divide large and complex rules into smaller parts, meaning there might be a sweet spot for the ideal division of statements into rules. Since our language allows the user so themselves decide where to create the divisions between states, it is possible to experiment with different divisions directly in the code.

\subsection{Translation of Individual Statements} \label{subsec:translation_statements}

All callable functions are implemented as both an executable implementation and a Tamarin translation function. The translation function defines how a statement in the AutoTam language should be integrated into a Tamarin rule. Such a translation function may either syntactically replace the output variables with expressions, introduce an \texttt{Out} fact in the case of a send operation, introduce an \texttt{In} fact in the case of a receive or user input function, introduce an \texttt{Fr} fact for a newly generated unique number, or introduce an equality restriction action fact \texttt{Eq} or \texttt{Neq}. 

We show a list of representative translations in Table~\ref{tab:dsl_small_translation}. A list of additional available translations can be found in Table~\ref{tab:dsl_translation} in the Appendix. The type of translation meant is indicated in brackets. [SET] type will replace occurrences of the left-hand-side variable to the right-hand-side expression. The [ACTION] type adds an action fact, typically used for equality restrictions. The [IN] and [OUT] type adds facts to the precondition and postcondition, respectively, of a Tamarin rule.
For example, the implementation of the function \textquote{EC25519-DH} computes the Elliptic-Curve group element from a secret x and a public key gy and stores it in variable gxy, while the translation function syntactically replaces all occurrences of gxy in the rule with the expression gy\texttt{\^}x which is Tamarin syntax for Diffie-Hellman group exponentiation.

\begin{table}
    \centering
    \caption{Example Translations. The complete translation table can be found in Appendix.}
    \begin{tabular}{l l}
        \toprule
        \textbf{Function call (AutoTam)} & \textbf{Translation (Tamarin)} \\
        \midrule
        gx = EC25519PublicKey(x) & [SET] gx $\rightarrow$ 'g'\texttt{\^}x \\
        gxy = EC25519DH(x, gy) & [SET] gxy $\rightarrow$ gy\texttt{\^}x \\
        a = Blake2s\_hash(b) & [SET] a $\rightarrow$ h(b) \\
        x = GenSecret(SIZE) & [IN] Fr($\sim$x) \\
		m = recv(fd, SIZE) & [IN] In(m) \\
		send(fd, m) & [OUT] Out(m) \\
        a = concat(b, c) & [SET] a $\rightarrow$ $\langle$b, c$\rangle$ \\
        t = timestamp() & [SET] t $\rightarrow$ \$t \\
        RequireEq(a, b) & [ACTION] Eq(a, b) \\
        RequireNEq(a, b) & [ACTION] Neq(a, b) \\
        x = Arg(SIZE) & [IN] In(x) \\
        x = Arg\_priv\_dh(SIZE) & [IN] !DHLtk($\sim$x) \\
        x = Arg\_pub\_dh(SIZE) & [IN] !DHPk(x) \\
        \bottomrule
    \end{tabular}
    \label{tab:dsl_small_translation}
\end{table}

To ensure that unique terms have new names in the Tamarin model, the statements to be translated are first transformed to single static assignment form~\cite{cytron1989ssa}. Since there is no control flow within each extended transition, the translation is trivial.

\subsection{Soundness of Translation}
\label{subsec:soundness_translation}

We present a proof of the soundness of translation with respect to trace properties. 
The desired soundness property can be stated in terms of trace inclusion. Informally, any global trace of some protocol participants executing the protocol must have a corresponding trace in the generated Tamarin model. Then, any proof in Tamarin stating that a given trace property holds for every trace in the model, proves that an equivalent trace property holds for any scenario during concrete execution.

The basic proof idea is to prove that a global execution trace of completed transitions matches a Tamarin trace where the corresponding rule for each transition is applied. We prove this by proving that the set of transition traces is preserved when performing the various transformations described in Section~\ref{subsec:translation_rules}.

We make our proof in a symbolic model, treating program variables as abstract terms rather than bitstrings. We assume a symbolic adversary model, both at the level of our execution semantics and the semantics of the translated model. A treatment on the computational soundness of a symbolic model is out of scope for this work. The survey by Cortier et al.~\cite{cortier2011surveysymbolic} provides an overview of works on bridging the gap between symbolic and computational protocol models. 

\begin{table}[]
    \centering
    \caption{Notation}
    \begin{tabular}{l c|l c}
        \toprule
        $i$ & Statement  & $P$ & protocol role description \\
        $S$ & State & $\Pi$ & Scenarios of Descriptions \\
        $T$ & Transition & $x$ & Execution State \\
        $p$ & Participant & $X$ & Global Execution State \\
        $C$ & Scenario & $w_e$ & Execution Trace \\ 
        $R$ & Rule  & $w_r$ & Rule Trace \\ 
        $F$ & Translation & $\tau_e$ & Traces of Scenario \\ 
        $M$ & Rule Mapping & $\Psi$ & Transition Sequence \\ 
         &  & $\pi$ & Traces of Rules \\ 
    \end{tabular}
    \label{tab:autotam_notation}
\end{table}

We introduce some notation as we formalise the execution semantics of AutoTam protocol role descriptions. We summarise the main notation used in Table~\ref{tab:autotam_notation}.
An AutoTam protocol role description consists of states and transitions. A state $S$ consists of two sequences of statements $S = ( [ i_{O1}, i_{O2}, ..., i_{On} ], [ i_{I1}, i_{I2}, ..., i_{Im} ] )$, where the first sequence is the output part and the second is the input part. A transition $T_{ab}$ between states $S_a$ and $S_b$ consists of a preceding state, a sequence of statements, and a succeeding state $T_{ab} = ( S_a, [ i_1, i_2, ..., i_n ], S_b )$. A protocol role description $P$ consists of a set of states and a set of transitions $P = ( \{ S_I, S_x, S_y, S_z, ... \}, \{ T_{Ix}, T_{xy}, ... \} )$. We assume all our protocol role descriptions are syntactically valid and follow Requirement~\ref{req:no_live_tran} and Requirement~\ref{req:no_stop_out}.

We introduce the notion of a participant, as an instance of an executing protocol role description, given globally unique values for all private keys, and knowledge of public keys of other participants. We define a scenario $C$ for a set of protocol role descriptions as a group of participants,
and denote the set of all scenarios for a given set of descriptions as $\Pi$. $C = \{p_{a}^1, p_{b}^1, p_{c}^2, ...\} \in \Pi(\{P_1, P_2, ...\})$, where a participant $p_n^m$ is an instance of protocol role description $P_m$. We place no bound on the number of participants for any protocol role description. 

An execution state $x_a$ of a participant $p_a$ is given by a statement index $pc$ within either a state or transition $B$, along with a set of terms $V$, each mapped to a variable in the protocol role description, local to the participant: $x_a = (pc, B, V)$. 

A global execution state for a scenario is given by the collective execution state of all participants, for example $X_1 = (x_a^1, x_b^1, x_c^1)$. We define a global execution trace $w_e$ as a sequence of global execution states starting from the initial state, where each next global execution state has executed one statement for one of the participants. We assume without loss of generality that concurrent executions appear in the global trace in some arbitrary order. We denote the set of admissible global traces for a scenario $C$ as $\tau_e(C)$. We write traces using arrow notation: $w_e = X_1 \rightarrow X_2 \rightarrow X_3 \in \tau_e(C)$.

We define the \emph{transition sequence} as the sequence of \emph{fully completed} transitions performed by a participant. We define a \emph{global transition sequence} as the transition sequence of the global execution trace. We define a mapping $\Psi$ from an execution trace to its transition sequence, consisting of a sequence of all transitions in the trace which were executed completely.

Before presenting our soundness theorem and its proof, we prove lemmas with respect to the various translations made by AutoTam.
We first define a \emph{transition-truncated global execution trace} as a global execution trace with no partially completed transitions.

\begin{lemma} \label{lemma:trunc_transitions}
For any scenario $C$ and any execution trace $w_e \in \tau_e(C)$, there exists a transition-truncated global execution trace $w_e* \in \tau_e(C)$ with the same global transition sequence $\Psi(w_e) = \Psi(w_e*)$.
\end{lemma}

\begin{proof}
We depend on Requirement~\ref{req:no_live_tran}.
Assume the lemma is false. Then there must exist some scenario $C$ with a trace $w_e$, that has no global trace $w_e^* \in \tau_e(C)$ without partially completed transitions so that $\Psi(w_e) = \Psi(w_e^*)$. There must then exist some transition $T_{xy}$ performed in $w_e$ which was not possible without partial execution of some transition $T_{xz}$. This implies that an executed statement in $T_{xz}$ affected future execution of the program, or another participant. This contradicts Requirement~\ref{req:no_live_tran}, and our requirement that no network-visible execution may occur during a transition.
\end{proof}

We define a \emph{complete global execution trace} as a global execution trace with no partially completed transitions or states, where all participants complete the output part upon completing the preceding transition, and no input parts are completed without some following transition being completed.

\begin{lemma} \label{lemma:streamline_state}
For any scenario $C$ and any transition-truncated global execution traces $w_e^* \in \tau_e(C)$, there exists a complete global execution trace $w_e^\# \in \tau_e(C)$ with the same global transition sequence $\Psi(w_e^*) = \Psi(w_e^\#)$.
\end{lemma}

\begin{proof}
We give the full proof of Lemma 2 in Section~\ref{sec:proof_lemma2} in the Appendix. In short, we prove that removing input parts or adding output parts to the execution cannot reduce the set of traces, leveraging Requirement~\ref{req:no_stop_out} to guarantee that execution of the output part is possible.
\end{proof}

\begin{lemma} \label{lemma:nondeterministic}
The set of global traces for a scenario when participants follow a deterministic transition order is a subset of the set of global traces for the same scenario when participants follow a non-deterministic transition order. 
\end{lemma}

\begin{proof}
The addition of non-deterministic choice can only add additional possibilities, since the previous traces remain possible.
\end{proof}

We define our translation function $F(U) = (R, M)$ for generating a Tamarin model using AutoTam as a function which takes a set of protocol role descriptions $U$ and outputs a set of rules $R = \{K_s, r_{ab}, r_{ab}, ...\}$ and a mapping function $M$. $R$ contains a rule for each transition which describes the cumulative effect of performing a complete extended transition, meaning a joined, output-transition-input sequence of statements, in addition to a key-share rule $K_s$ which provides unique private keys for each role. The trace mapping function $M$ maps each rule in a rule trace to its transition, or skips if the rule is $K_s$. For example, $M( K_s \rightarrow r_{ab}  \rightarrow K_s \rightarrow r_{bc} \rightarrow  ...) =  T_{ab}  \rightarrow  T_{bc} \rightarrow  ...$. We write $\pi(R)$ for the set of admissible rule traces of $R$.

We now state our soundness theorem for the translation.

\begin{figure}
    \centering
    \includegraphics[width=0.9\columnwidth]{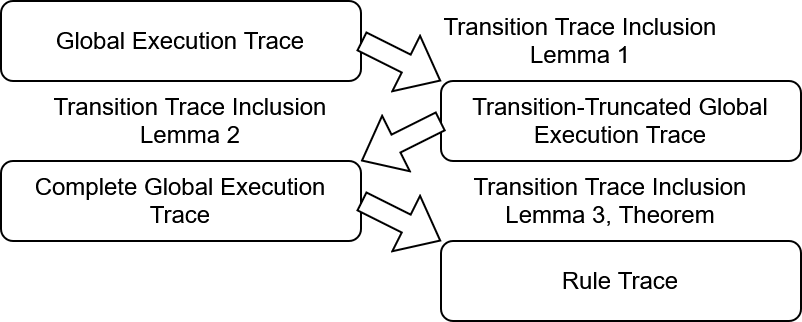}
    \caption{Transition Trace Inclusion Proof Steps.}
    \Description{A sequence of proof steps, with labelled arrows between. Arrow n is directed arrow from Box n to Box n+1. Box 1: Global Execution Trace; Arrow 1: Transition Trace Inclusion (Lemma 1); Box 2: Transition-Truncated Global Execution Trace; Arrow 2: Transition Trace Inclusion (Lemma 2); Box 3: Complete Global Execution Trace; Arrow 3: Transition Trace Inclusion (Lemma 3, Theorem); Box 4: Rule Trace.}
    \label{fig:autotam_proof}
\end{figure}

\begin{theorem}
For any set of protocol role descriptions $U$ translated as $F(U) = (R,M)$, then for any scenario $C \in \Pi(U)$ all global execution traces $w_e \in \tau_e(C)$ have a matching trace $w_r \in \pi(R)$ so that $\Psi(w_e) = M(w_r)$
\end{theorem}

\begin{proof}
We illustrate the proof in Figure~\ref{fig:autotam_proof}.  Lemmas~\ref{lemma:trunc_transitions} and~\ref{lemma:streamline_state} together prove that for any execution trace $w_e \in \tau_e(C)$ there exists a global complete execution trace $w_e^\# \in \tau_e(C)$ so that $\Psi(w_e) = \Psi(w_e^\#)$. Lemma~\ref{lemma:nondeterministic} proves that $w_e^\#$ is a valid trace of $C$ even if executions of descriptions of $C$ allow nondeterministic choice of transition. 
From our definition of $F$ the rules $R$ describe complete extended transitions which perform the identical effect of a transition in a global complete execution trace with nondeterministic choice. Thus, there is a valid trace $w_r \in \pi(R)$ with $\Psi(w_e) = M(w_r)$.

\end{proof}

Some details are omitted from our proof. We do not explicitly formalise how variables are mapped in rewrite rules, instead the proof relies on variables being mapped one-to-one, while in practice the translation optimises the rewrite rules to only consider locally relevant variables. 
We also implicitly assume that the translation table, describing the mapping of statements to a symbolic Tamarin model according to Table~\ref{tab:dsl_translation}, soundly captures the behaviour of the concrete function implementations. We omit most of the details about key-share of pre-distributed long-term keys, which is needed for any participant to meaningfully execute the protocols, but is technically performed outside AutoTam execution. We assume the key-share rule $K_s$ models the key-share method. Finally, we omit the description of the network adversary, as we assume the adversary model described in Section~\ref{subsec:sym_protocol_verif} across our models.

\subsection{Built-In Features}

In addition to the translation described, there are several built-in features which assist in creating Tamarin models which are readable and efficient. Let-bindings are automatically used in place of expressions which are repeated across a generated rewrite rule, greatly improving readability. Lists are automatically unfolded in order to ensure soundness, since expressions such as $h(<a,<b,c>>)$ and $h(<<a,b>,c>)$ are treated as non-unifiable terms in Tamarin, despite both representing a hash applied to the same sequence of terms $a,b,c$. This can be remedied by choosing to use either a left-associative or right-associative syntax throughout a model. We automatically translate all expressions of either form into fully unfolded form to improve readability, e.g. $<a,b,c>$, which in Tamarin is syntactic sugar for right-associative pairs.

User provided private keys are automatically typed as private terms in the generated Tamarin model. Some functions, such as computing a timestamp, are modelled as public terms in place of timestamps. Typing of terms is automatically and soundly propagated across rewrite rules, by ensuring that the term has the given type in all rules which produce the incoming state fact. 

AutoTam additionally allows automatic generation of Tamarin lemmas for standard security properties such as injective agreement from Lowe's hierarchy~\cite{Lowe1997AuthenticationHierarchy} and weak secrecy. The user only needs to specify the name of the transitions which are to be used as events in building the properties, and the names of agreement variables.

\section{Symbolic Execution of Implementations} \label{sec:symbolic_execution_method}

To gain additional confidence in the security of implementations, and to decreases the reliance on the security of our interpreter, we integrate dynamic symbolic execution into AutoTam.
We utilise the KLEE symbolic execution engine~\cite{cadar2008kleecoverage}, and follow the method of Wilson and Asplund~\cite{KleeNordsec} for performing efficient symbolic execution of protocol implementations using KLEE. The method consists of replacing the concrete implementations of cryptographic functions with symbolic function abstractions, treating network input as fully symbolic, and delaying concretization of message lengths as far as possible during analysis. In addition, the common \textquote{memcmp} function from the standard~C library is given a special function abstraction. Through this method, it is possible to prove that an implementation does not access memory outside of the largest symbolic length of buffers when executed through our interpreter, ruling out severe memory errors. Moreover, symbolic execution can be extended to support other analyses such as conformance tests, see e.g.~\cite{Asadian2024ares, chau2019SymCert2, Vanhoef2018symbolicexecutionprotocols}.

We have implemented function abstractions for all cryptographic functions available in our language. We use function abstractions which return symbolic data instead of performing cryptographic operations. This way we may still explore the behaviour assuming any output from cryptographic functions is possible. As discussed in Section~\ref{subsec:language_variables}, we rely on the formally verified HACL library for cryptographic operations.

Through this method we can systematically, and in the ideal case completely, explore the behaviour of our implementations when executed using our interpreter. Since at any point the program would read network input we produce a new symbolic message, we gain complete coverage with respect to arbitrary network input. 

\section{Case Studies} \label{sec:AutoTam_casestudies}

We perform two case studies.
One on a signed Diffie-Hellman key-exchange protocol, and one on the WireGuard VPN protocol~\cite{donenfeld2017ndsswireguard}. 
We compare our generated WireGuard model with existing Tamarin models from the literature. 
We additionally conduct a performance evaluation of our WireGuard implementation. We measure throughput when running our implementation using our interpreter and compare to two existing WireGuard implementations: the official Linux kernel implementation and the official user-level Go implementation. Our WireGuard implementation is interoperable with both implementations, although some features are omitted from our implementation such as DoS mitigation and timed renegotiation.

\subsection{A Diffie-Hellman Key-exchange Protocol}

We study a Diffie-Hellman key-exchange (DH) protocol designed to establish a new ephemeral key between two endpoints which have known long-term signing keys. We follow a version of DH similar to the one used in a case study by Arquint et al.~\cite{arquint2023tamarintoimpl}. The established keys are authenticated through a signature using the long-term keys. The protocol is shown in Alice and Bob notation in Table~\ref{tab:dhprotocol}, where roles are shown to the left, with the arrow denoting the direction of the message sent going from sender to receiver. To the right of the colon the message fields are shown, with angle brackets denoting list structures, quoted characters referring to constants, and exponentiation representing Diffie-Hellman public key operations. We use curly brackets for signatures or encryption.

The protocol has two roles, labelled A for the initiator role and B for the responder role. The protocol starts with an agent in the initiator role sending a plaintext message containing a fix header '1', its own unique 64-bit identifier I, the 64-bit identifier of another agent R, and a new ephemeral Diffie-Hellman public key $g^x$. A responder which receives a valid message will compute its own public key $g^y$ and sign a response message with the header '2', the identifiers, and the public keys using its own long-term key ltkB. The header '10', the message, and the computed signature is sent as a response. Finally, the Initiator will compute its own signed response message but with the header '3' in similar fashion. Upon completion both agents may compute the same value of $g^{xy}$, completing the key-exchange.

\subsection{The WireGuard VPN Protocol}

WireGuard is a protocol designed to provide encryption for IP packets suitable for secure VPN services, similar to IPsec or OpenVPN. The WireGuard protocol operates on top of UDP, where the IP messages to send are encrypted and appended to a short message header. To establish the keys to use for encryption, a handshake is first performed. The handshake is based on the IKpsk2 protocol from the Noise framework~\cite{perrin2018noise}. The protocol is initiated by the sender of the first message, and proceeds with a response from the responder, after which the encrypted transport messages may be sent, first by the initiator and then from either endpoint. The protocol is shown in simplified form in Table~\ref{tab:wgprotocol}.

We use curly brackets to denote AEAD encryption using derived keys.
The encryption and decryption keys are derived in a long chain of hashes, Diffie-Hellman operations, and hash-based key-derivation functions (HKDF) which we omit from the notation. To highlight the complexity of the key derivation, we note that deriving the encryption keys and first mac in the first message from the initiator requires thirteen separate hashing operations when the HKDF is expanded. 
The final message is the transport message which encrypts an IP packet \texttt{msgN}, and can be repeated, increasing the counter \texttt{csend}.

\begin{table}
    \centering
    \caption{Diffie-Hellman Key-exchange Protocol with revealing signature.}
    \label{tab:dhprotocol}
    \begin{tabular}{l}
        \toprule
        A $\rightarrow$ B: $\langle$'1', I, R, $g^x$$\rangle$ \\
        B $\rightarrow$ A: $\langle$'10', sign\{$\langle$'2', R, I, $g^x$, $g^y$$\rangle$\}$_{ltkB}\rangle$ \\
        A $\rightarrow$ B: $\langle$'10', sign\{$\langle$'3', I, R, $g^x$, $g^y$$\rangle$\}$_{ltkA}\rangle$
    \end{tabular}    
\end{table}

\begin{table}
    \centering
    \caption{WireGuard Protocol}
    \label{tab:wgprotocol}
    \begin{tabular}{l}
        \toprule
        A $\rightarrow$ B: $\langle$'1', '000', SidA, $g^x$, \{$g^{ltkA}$\}, \{$ts$\}, mac1, mac2$\rangle$ \\
        B $\rightarrow$ A: $\langle$'2', '000', SidB, SidA, $g^y$, \{\_\}, mac1, mac2$\rangle$ \\
        A $\leftrightarrow$ B: $\langle$'4', '000', receiverSid, csend, \{msgN\}$\rangle$
    \end{tabular}
\end{table}

\subsection{Implementation and Verification Effort}

The sizes of our AutoTam language protocol role descriptions are shown in Table~\ref{tab:impl_effort}. The protocol role descriptions include parsing and serializing behaviour for both roles. Incoming network messages are parsed and serialized using the \texttt{concat} and \texttt{split} functions for concatenating and dividing parts of messages, respectively. Most of the parsing behaviour can be implemented by dividing a message into appropriately sized parts and comparing each part with a constant or known value. We do not implement DoS mitigation, the moving window for avoiding duplicate counter values, and concurrent reading and writing of transport messages in our WireGuard implementation. 

For both protocols, we generate a Tamarin model. The generated model includes lemmas generated for both roles in each protocol. We produce lemmas for the handshake phase of both protocols, including weak secrecy of derived keys, and authentication as injective agreement on derived keys. We generate a reachability lemma showing that a handshake can complete for both protocols as a sanity check. This ensures that the other lemmas do not vacuously hold. The final row of Table~\ref{tab:impl_effort} shows the size of the generated Tamarin model.
The time to generate Tamarin models including lemmas for both of our test examples is less than 1s. We verify our models using Tamarin on a 16Gb RAM Intel Xeon 3.5Ghz CPU server.
The model for the DH protocol verifies without any modification for all security properties in less than 5s using Tamarin. 

\begin{table}[]
    \centering
    \caption{Implementation Size of DH and WireGuard Protocols and Size of Generated Tamarin Model.}
    \begin{tabular}{l c c}
        \toprule
        & DH & WireGuard \\
        \midrule
        \textbf{LoC Initiator} & 68 & 222 \\
        \textbf{LoC Responder} & 74 & 243 \\
        \midrule
        \textbf{Generated Lines} & 107 & 242 \\
        \bottomrule
    \end{tabular}
    \label{tab:impl_effort}
\end{table}

\begin{table}
    \centering
    \caption{Verification Results for Tamarin Models}
    \begin{tabular}{l c c}
        \toprule
         & DH & WireGuard \\
        \midrule
        \textbf{Time} & 3.06s & 359.62s \\
        \textbf{Reachability Initiator} & $\checked$ & $\checked$ \\ 
        \textbf{Weak Secrecy Initiator} & $\checked$ & $\checked$ \\
        \textbf{Injective Agreement Initiator} & $\checked$ & $\checked$ \\
        \textbf{Reachability Responder} & $\checked$ & $\checked$ \\ 
        \textbf{Weak Secrecy Responder} & $\checked$ & $\checked$ \\
        \textbf{Non-Injective Agreement Resp.} & $\checked$ & $\checked$ \\
        \textbf{Injective Agreement Responder} & $\checked$ & N/A* \\
        \bottomrule
        \multicolumn{3}{l}{*Not applicable for the handshake phase.}
    \end{tabular}
    \label{tab:Tamarin_Lemmares}
\end{table}

Verification of the WireGuard protocol is more challenging. In order to verify the generated model, we were required to manually provide a so-called sources lemma (see e.g.~\cite{cortier2022automaticsource}). WireGuard being challenging to verify is in line with other works on the verification of WireGuard using Tamarin, as the more detailed verified models of the protocol have required some form of manual guidance to verify efficiently, either in the form of writing sources lemmas~\cite{arquint2023tamarintoimpl}, writing an oracle~\cite{arquint2023tamarintoimpl, Girol2020noisetamarin, lafourcade2024unifiedwireguard}, or providing state invariants for the key-exchange~\cite{donenfeld2017formalwireguard}. 
Writing the sources lemma for the generated model was not any more challenging than writing a source lemma for any other model. We used the interactive mode of the Tamarin prover to identify the missing partial deconstructions and their cause. After some experimentation we found it sufficient to prove that adversary knowledge of the terms \texttt{sock}, \texttt{src\_addr} and \texttt{r\_sender\_index} was always preceded by the fourth action fact of the initiator role.

We show the verification results for the automatically generated DH model and the WireGuard model with the added sources lemma in Table~\ref{tab:Tamarin_Lemmares}.

\subsection{Comparing Generated and Expert Models} \label{subsec:wireguard_tamarin_compare}

\begin{table*}[h]
    \centering
    \caption{Level of detail in resulting Tamarin Models compared to other verified WireGuard models. We use slashes to denote partially modelled features. Lines of code (LoC) does not count lemmas.}
    \begin{tabular}{c c c c c c c c}
        \toprule
        \textbf{Model} & \textbf{Protocol} & \textbf{Executable} & \textbf{MAC} & \textbf{Decryption} & \textbf{Transport} & \textbf{HKDF} & \textbf{LoC} \\
        \midrule
        Donenfeld and Milner \cite{donenfeld2017formalwireguard} & WireGuard & X & X & no nonce & $\diagup$* & X & 225 \\
        Girol et al. \cite{Girol2020noisetamarin} & Noise & X & X & $\diagup^\dagger$ & $\checked$ & X & 311 \\ 
        Arquint et al. \cite{arquint2023tamarintoimpl} & WireGuard & $\diagup$** & X & X & $\checked$ & X & 293 \\
        Lafourcade et al. \cite{lafourcade2024unifiedwireguard} & WireGuard & X & $\checked$ & $\diagup^\dagger$ & $\checked$ & X & 371 \\
        AutoTam & WireGuard & $\checked$ & $\diagup^\ddagger$ & $\checked$ & $\diagup^\#$ & $\checked$ & 180 \\
        \bottomrule
		\multicolumn{3}{l}{\footnotesize{*Only the first transport message is modelled.}} 
        & &
        \multicolumn{3}{l}{\footnotesize{$^\dagger$Combined ciphertext and MAC.}} \\
        \multicolumn{3}{l}{\footnotesize{**Manually written protocol implementation.}}
        & &
        \multicolumn{3}{l}{\footnotesize{$^\#$Modelled but no property verified for transport.}} \\
        \multicolumn{6}{l}{\footnotesize{$^\ddagger$Only MAC1 is modelled, as MAC2 is only used for DDoS mitigation which is not implemented.}} \\
        
    \end{tabular}
    \label{tab:Tamarin_Wireguard}
\end{table*}

We compare our generated WireGuard model to other works which also use Tamarin for the verification of the WireGuard protocol. 
While other works have performed symbolic analysis for more advanced security properties such as verifying the secrecy of transport message payloads and anonymity, the security properties are not our main focus, as only basic authentication and secrecy properties are currently supported for automatic generation by our tool. The comparison focuses on the fidelity 
of models, meaning the amount of detail kept from the protocol role description.

We show our comparison in Table~\ref{tab:Tamarin_Wireguard}.
The columns indicate that a feature is fully modelled as authentically as possible using a checkmark, or that a feature is missing from the model using X. 
The Protocol column indicates whether the model follows the WireGuard message format or the Noise IKpsk2 message format. The MAC column indicates whether both MAC fields of the message are authentically modelled as derived from message hashes. Here the second MAC is used only for DoS mitigation, which we have not implemented and which is therefore not modelled in our model. The Decryption column indicates whether AEAD was authentically modelled with both a decryption and validation of ciphertext. 

No Tamarin model explicitly describes the steps performed by the key derivation functions as individual hash operations performed by WireGuard except for our model. While not strictly needed for soundness of verification, the expanded HKDF operations showcase a level of detail in our generated model which is cumbersome to model by hand.
The Tamarin model by Donenfeld and Milner~\cite{donenfeld2017formalwireguard} uses static strings in place of both MAC fields, and includes the first but not subsequent transport messages. Girol et al.~\cite{Girol2020noisetamarin} verify the Noise protocols, which includes the IKpsk2 protocol. 
IKpsk2 follows a slightly different format which does not contain some fields such as the MACs. The model by Arquint et al.~\cite{arquint2023tamarintoimpl} was made with the intention of generating requirements in separation logic. The reparation logic served as the specification when verifying the WireGuard implementation. This model did not authentically model the decryption of message terms, instead using encryption both when sending a message and when receiving a message. While the verification is still sound, it does not accurately represent the way an implementation would perform decryption.
The most detailed model of WireGuard that we are aware of is by Lafourcade et al.~\cite{lafourcade2024unifiedwireguard} who wrote their models using SAPIC+~\cite{cheval2022sapic}, a tool which can generate models for several different protocol verifiers including Tamarin. They analyse several scenarios for key compromise, including leakage of precomputed static keys and ephemeral keys, and their models include the complete transport phase. They additionally verify a number of different compromise scenarios, such as static ephemeral key leakage. They used a dedicated 256 core server for the verification of their models, where verification of all lemmas took at most 4h and 45min when using Tamarin. 

\subsection{Symbolic Execution Analysis}

We verify the safety of the executable protocol implementations using the method of symbolic execution described in Section~\ref{sec:symbolic_execution_method}. We show the results of applying the method to both roles of our DH and WireGuard implementations in Table~\ref{tab:sym_execution_res}. For the WireGuard protocol, the number of transport messages is limited to a single message.

\begin{table}
    \caption{Results of Symbolic Execution Analysis, showing analysis time, number of paths explored by KLEE, and the number of exceptions encountered.}
    \label{tab:sym_execution_res}
    \begin{tabular}{c c c c}
        \toprule
        \textbf{Implementation} & \textbf{Time} & \textbf{\# Paths} & \textbf{Errors} \\
        \midrule
        DH initiator & 59s & 965 & 0\\
        DH responder & 1min 22s & 1262 & 0\\
        WireGuard initiator & 21s & 44 & 0 \\
        WireGuard responder & 16min 33s & 4120 & 0\\
        \bottomrule
    \end{tabular}
\end{table}

Our analysis rules out most types of memory errors automatically. The analysis time and number of paths is highly dependent on the parsing behaviour of each role. The initiator side of the WireGuard protocol is the fastest to analyse, as it only parses a single relatively short message.
Analysis of the responder size of the WireGuard protocol takes considerably longer and generates a significantly larger number of execution paths. This is because this side of our implementation manages variably sized messages, in the form of transport messages with arbitrary encrypted data. Since our implementation limits the packet size to at most 512 bytes, the search space is still finite, and analysis can still exhaustively search through every possible message length.

\subsection{Protocol Performance}

We set up an experiment to evaluate the performance of our WireGuard implementation when using our AutoTam interpreter. We evaluate the throughput when connecting our client and server implementations over a 1000Mb/s ethernet connection, each implementation running on a 16Gb RAM Intel Xeon 3.5Ghz CPU server. 
We perform the experiment by measuring the time it took to receive 200 million fix-size UDP messages with a size of 448 bytes in UDP data from the initiator to the responder, starting from the time that the first message was received. Throughput was then computed as number of bits received over time taken.  

We show the setup of our experiments in \figurename~\ref{fig:autotam_eval}.
When measuring the throughput for the Linux kernel and Go implementations, we used a UDP client and server implementation where the client implementation sent UDP messages to the WireGuard routed address as fast as possible. 
Since our implementation is not a full TUN device implementation, we set up the implementation to itself generate the UDP messages to transport using WireGuard, while on the receiver end no further routing was performed and we instead timed receiving the WireGuard transport message payload. Because of this, the comparison is not fully fair, as there is some additional routing necessary to read and write messages between a UDP client and WireGuard implementation.

\begin{figure}
    \centering
    \includegraphics[width=\linewidth]{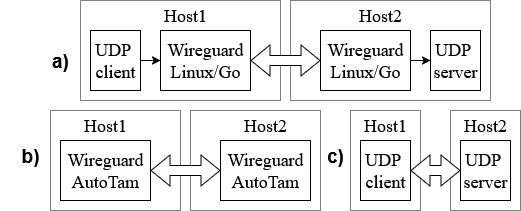}
    \caption{Evaluation Setups. a) Measuring throughput of Linux and Go implementations. b) Measuring throughput of AutoTam implementation. c) Measuring throughput with no WireGuard routing.}
    \Description{Three evaluation setups, labelled A, B and C, are shown as a series of connected boxes. Setup A contains two hosts, host 1 containing a UDP client sending to a WireGuard Linux/Go box, which is connected both ways to an identical box on host 2. The WireGuard Linux/Go box on host 2 sends to a UDP server. Setup B contains two hosts, each with a box labelled WireGuard AutoTam connected to each other. Setup C contains two hosts, host 1 labelled UDP client, host 2 labelled UDP server, both connected to each other.}
    \label{fig:autotam_eval}
\end{figure}

The results of the performance evaluation can be found in \figurename~\ref{fig:throughput}. Our implementation performs at roughly the same throughput as the official WireGuard Go implementation when ran by our interpreter. 
Considering that we have not made any specific optimizations as part of our interpreter implementation, it seems likely that further optimizations are possible.

\subsection{Interoperability}

We tested running our implementation alongside both the WireGuard Go and Linux kernel implementation in order to ensure our implementation is interoperable. We ensure our implementation can perform both sides of the WireGuard handshake, using two experiments.
For the first scenario, we ensure that the UDP message send by the UDP client results in a completed handshake between the given WireGuard implementation and our AutoTam implementation, and that the UDP message gets delivered. For the second scenario, we ensure that the handshake initiated by the AutoTam implementation is accepted by the given WireGuard implementation.

\section{Related Work} \label{sec:AutoTam_relatedwork}

\begin{figure}
    \centering
    \includegraphics[width=\linewidth]{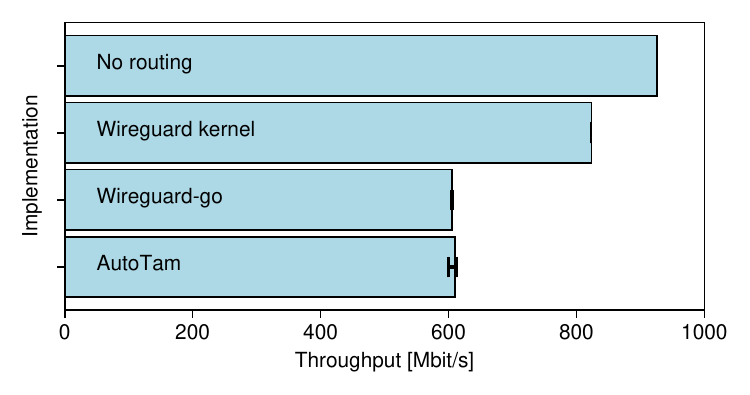}
    \caption{Measured Throughput of WireGuard Implementations, showing median of three experiments.}
    \Description{Barplot with four bars: 'No routing' at around 930Mbit/s, 'Wireguard kernel' at 820Mbit/s, 'Wireguard-go' at 600Mbit/s, and 'AutoTam' at 600Mbit/s.}
    \label{fig:throughput}
\end{figure}

\begin{table*}
    \centering
    \caption{Tools and Methods for Formal Verification of Protocol Implementations.}
    \begin{tabular}{l c c c c c c c}
        \toprule
        \textbf{Work} & \textbf{Type} & \textbf{Source} & \textbf{Trace Property} & \textbf{Impl.} & \textbf{Verif.} & \textbf{Loops} & \textbf{Largest} \\
        & & \textbf{$\rightarrow$Target} & \textbf{Proof Method} & \textbf{Expertise} & \textbf{Expertise} & & \textbf{Case-Study} \\
        \toprule
        Expi2Java~\cite{Backes2012expi2java} & 1 & Expi$\rightarrow$Java & ProVerif & \RIGHTcircle & \RIGHTcircle & Y & TLS 1.0 \\
        Cadé et al.~\cite{cade2012computationally, cade2015proved} & 1 & CryptoVerif$\rightarrow$OCaml & CryptoVerif & \RIGHTcircle & \RIGHTcircle & N & SSH \\
        Arquint et al.~\cite{arquint2023tamarintoimpl} & 1 & Tamarin$\rightarrow$Java/Go & Tamarin & \RIGHTcircle & \RIGHTcircle & Y & WireGuard \\
        Owl/OwlC~\cite{Gancher2023owl, singh2025owlc} & 1 & Owl$\rightarrow$Rust & Owl & \RIGHTcircle & \RIGHTcircle & Y & WireGuard \\
        \midrule
        CSur~\cite{GoubaultLarrecq2005csur} & 2 & C$\rightarrow$SMT & Horn clauses & \RIGHTcircle & \RIGHTcircle & Y & Needham-Schroeder \\
        ASPIER~\cite{chaki2009aspier} & 2 & C$\rightarrow$COPPER & model checker & \Circle & \RIGHTcircle  & Y & OpenSSL core \\
        c2pv~\cite{aizatulin2011cryptosymbolicexec} & 2 & C$\rightarrow$ProVerif & ProVerif & \Circle & \RIGHTcircle & N & minexplib \\
        CryptoBap~\cite{Nasrabadi2023CryptoBap, nasrabadi2025automated} & 2 & Binary$\rightarrow$SAPIC+ & SAPIC+~\cite{cheval2022sapic} & \Circle & \RIGHTcircle & Y & Whatsapp core \\
        \midrule
        VCC~\cite{dupressoir2014guiding} & 3 & C & Coq & \Circle & \CIRCLE & Y & Otway-Rees \\
        Igloo~\cite{sprenger2020igloo, pereira2025router} & 3 & Any & Isabelle / HOL & \Circle & \CIRCLE  & Y & SCION Router \\
        Arquint et al.~\cite{Arquint2023modular} & 3 & Any & Separation Logic & \Circle & \CIRCLE & Y & WireGuard \\
        \midrule
        F7~\cite{bengtson2011refinement, bhargavan2013miTLS} & 4 & F\# & Typecheck & \RIGHTcircle & \CIRCLE  & Y & TLS 1.2 \\
        DY*~\cite{bhargavan2021dystar, bhargavan2024dystar} & 4 & F* & Typecheck & \RIGHTcircle & \CIRCLE & Y & Signal / BA \\
        \midrule
        fs2pv~\cite{bhargavan2008fsquaretoproverif, bhargavan2006verifiedws} & 5 & F\#$\rightarrow$ProVerif & ProVerif & \Circle & \RIGHTcircle & Y & WS / Otway-Rees \\
        DJS~\cite{bhargavan2012defensivejs} & 5 & JS$\rightarrow$ProVerif & ProVerif & \Circle & \RIGHTcircle & Y & HMAC (50 LoC) \\
        AnBx~\cite{Modesti2014anbx2java} & 5 & AnBx$\rightarrow$Java & OFMC~\cite{basin2005ofmc} & \Circle & \RIGHTcircle  & N & 3KP \\
        SPS~\cite{Almousa2015SPS} & 5 & SPS$\rightarrow$JS,ProVerif, & ProVerif, & \Circle & \RIGHTcircle & N & TLS 1.2 \\
        & & ASLAN~\cite{armando2012avantssar} & ASLAN & & & & \\
        ProScript~\cite{Kobeissi2017ProScript, bhargavan2017TLSproverif} & 5 & JS$\rightarrow$ProVerif,CryptoVerif & ProVerif & \Circle & \RIGHTcircle & N & Signal / TLS 1.3 core \\
        hax~\cite{bhargavan2024hax, bhargavan2025protocolrust} & 5 & Rust$\rightarrow$F*,ProVerif, & ProVerif & \RIGHTcircle & \RIGHTcircle & Y & TLS 1.3 \\
        & & CryptoVerif & & & & & \\
        \midrule
        \textbf{AutoTam} & 5 & AutoTam spec.$\rightarrow$Tamarin & Tamarin & \Circle & \RIGHTcircle & Y & WireGuard \\
        \bottomrule
        \multicolumn{8}{l}{\textbf{Types:} 1 - Model To Implementation; 2 - Static Analysis Extraction; 3 - Direct Verification; 4 - Verification through Type Theory;} \\
        \multicolumn{8}{l}{5 - Implementation or Specification to Model. \textbf{Implementation Expertise:} \Circle Target/Source Language Only; \RIGHTcircle Some Formal Expertise} \\
        \multicolumn{8}{l}{\textbf{Verification Expertise:} \RIGHTcircle Model Language Expertise; \CIRCLE Manual Proof} \\
    \end{tabular}
    \label{tab:model_gen}
\end{table*} 

Prior to our work, we are not aware of any tool which can extract Tamarin models from executable implementations besides CryptoBap~\cite{Nasrabadi2023CryptoBap}. CryptoBap uses the HolBA binary analysis platform to symbolically execute protocol binaries for analysis, originally to ProVerif and CryptoVerif, with recent adaptions~\cite{nasrabadi2025automated} also targeting SAPIC+ which can translate to both Tamarin and DEEPSEC. Their goal is different from ours, as they are interested in performing thorough and deep analysis of existing implementations focusing on the protocol logic, while we are interested in verifying a complete reference implementation. Case studies reflect this, as they target the most relevant segments of binaries, extracting functions for handshake messages, while in contrast we analyse the entirety of protocol implementations written in the AutoTam language.

We categorise related works in Table~\ref{tab:model_gen}. We divide related work into five types, depending on the method used to connect between trace property verification and implementations. We list the input language and target language(s) of each tool in col.~3, and the verifier used for trace properties in col.~4.
We compare the needed formal expertise required for creating both the implementation itself in col.~5, and to perform the verification in col.~6. We believe no approach can be fully automatic in performing verification, since some verification domain knowledge will be needed to at least specify the properties to verify, identify protocol code sections for analysis, or to guide heuristics for more challenging protocols. 
Depending on if the approach can account for looping protocols, we write Y for yes and N for no in col.~7. Finally, we list the name of the largest case-study performed using the tool to our knowledge in col.~8.

Several works~\cite{Backes2012expi2java, cade2012computationally, cade2015proved, arquint2023tamarintoimpl, Gancher2023owl, singh2025owlc} (Type 1 in Table~\ref{tab:model_gen}) generate an implementation or contract from a verifiable protocol model, or a language close to the model. 
For this approach, the implementations must be made by formal verification experts knowledgeable in the modelling language.

A number of works including CryptoBap~\cite{GoubaultLarrecq2005csur, chaki2009aspier, aizatulin2011cryptosymbolicexec, Nasrabadi2023CryptoBap, nasrabadi2025automated} (Type 2 in Table~\ref{tab:model_gen}) develop methods for automatic verification of existing protocol implementations. This is important work in order to secure existing protocol implementations. However, the methods provide little guidance for how a new implementation for a new protocol should be written, and general static analysis methods rarely scale to entire protocol implementations. In the case of CSur~\cite{GoubaultLarrecq2005csur}, manual annotations are also required.

Some works~\cite{sprenger2020igloo, Arquint2023modular, pereira2025router} (Type 3 in Table~\ref{tab:model_gen}) design protocol implementations using sound refinement techniques from a verified formal model proven using a theorem prover. Others, e.g.~\cite{dupressoir2014guiding}, verify implementations directly using a theorem prover. For both approaches, the verification effort is significant. 

Other works~\cite{bengtson2011refinement, bhargavan2021dystar, bhargavan2024dystar} (Type 4 in Table~\ref{tab:model_gen}) encode the verification condition directly as type-check conditions in a type-theoretic language. Development in such languages requires good understanding of type theory and functional programming.
While specific to Noise protocols, Ho et al.~\cite{Ho2022noisestar} perform both symbolic verification and produce efficient implementations of the Noise protocols. Building proofs in F* was largely a manual effort, but once written, verified implementations for all noise protocols could be generated.

Finally, several works~\cite{bhargavan2008fsquaretoproverif, bhargavan2012defensivejs, Almousa2015SPS, Kobeissi2017ProScript, bhargavan2024hax, bhargavan2025protocolrust} (Type 5 in Table~\ref{tab:model_gen}) automatically translate implementations in high-level languages to a verifiable model. 
As discussed in the introduction, many of the existing works in this category do not cover the entire state machine of protocols, and in particular do not fully support looping protocol state machines. In the case of hax~\cite{bhargavan2024hax}, manual annotations are required to select functions to extract to ProVerif, and manual effort is needed to write the verification scenario.

Based on our classification, only five tools require no formal verification expertise during the implementation phase, does not require manual proofs, and supports looping protocol state machines: ASPIER, CryptoBap, fs2pv, DJS, and AutoTam. Out of these, all besides CryptoBap and AutoTam only verify small case studies to our knowledge.

In addition to the works already discussed in Subsection~\ref{subsec:wireguard_tamarin_compare}, there exists a number of works which analyse WireGuard which do not use Tamarin. A Computational proof was performed by Lipp et al.~\cite{lipp2019wireguardcryptoverif} using CryptoVerif. A manual proof of security in a computational model was performed by Dowling and Paterson~\cite{Dowling2018wireguardmanual}. Also, the closely related Noise protocol IKpsk2 has been analysed by Kobeissi et al.~\cite{Kobeissi2019noiseexplorer} and Ho et al.~\cite{Ho2022noisestar}. For a more detailed comparison of most of these works we refer to Lafourcade et al.~\cite{lafourcade2024unifiedwireguard}.

\section{Conclusions} \label{sec:AutoTam_conclusions_futurework}

In this paper, we have presented AutoTam, a domain-specific language and interpreter for cryptographic protocols, our AutoTam translator tool for soundly producing Tamarin models, and our integration of symbolic execution for protocol analysis of protocols written in AutoTam. 
The language is directly defined at the abstraction level of message terms, and as a result, implementation effort is focused on placing network operations, message parsing, and cryptographic operations into a state machine. Our sound translation enables verification of authenticity and secrecy properties at the implementation level.

We produce formally verified implementations of a signed Diffie-Hellman protocol and the WireGuard VPN protocol in our case-studies. We generate detailed symbolic models which accurately capture the semantics of the executable implementations, along with the verification goals in the form of high-level security properties. Since we use a sound transformation which retains trace properties between a set of AutoTam protocol role descriptions and their generated model, our proofs of the high-level properties formally verify the security of our implementations. Through symbolic execution we automatically rule out implementation-level memory errors. In addition, interoperability testing ensures that we implement and model the desired protocol.

\section*{Acknowledgments}

This work was partially supported by the Wallenberg AI, Autonomous Systems and Software Program (WASP) funded by the Knut and Alice Wallenberg Foundation.


{\footnotesize \bibliographystyle{acm}
\bibliography{ref}}

\appendix

\section{Open Science}

We provide the entire codebase of the implementation of AutoTam, including the AutoTam protocol role descriptions of our case studies. Our public repository can be found at~\url{https://gitlab.liu.se/ida-rtslab/public-code/autotam}. The project README explains how to build AutoTam on a Linux machine. The README additionally explains how to produce each artifact presented as part of our case studies. We have confirmed our project builds on Ubuntu 22.04.4 LTS. 

In the repository, we provide the artifact for the generated Tamarin model of the DH protocol under:
\begin{verbatim}
    examples/DH/automatic.spthy    
\end{verbatim}
The generated Tamarin model of the WireGuard protocol as originally generated can be found under:
\begin{verbatim}
    examples/Wireguard/automatic.spthy
\end{verbatim}
And the model with the added source lemma can be found as:
\begin{verbatim}
    examples/Wireguard/wg/auto_source.spthy
\end{verbatim}
The automatically generated files can be generated by building and running AutoTam as described in the README. 

In order to verify the models, a working installation of Tamarin is needed. Instructions for installing Tamarin are found at~\url{https://tamarin-prover.com/install.html}. We used Tamarin version 1.10.0 with Maude version 3.1. Once installed, automatic verification is performed by running:
\begin{verbatim}
    tamarin-prover [spthy file] --prove
\end{verbatim}
All lemmas should show as verified for both models.

The timing of the Tamarin verification, and the symbolic execution, were performed on an 16GB RAM Intel i7-1355U (12 core) computer.

To perform the benchmarking experiments, we used two 16Gb RAM Intel Xeon 3.5Ghz CPU (4 cores) servers running Ubuntu 22.04.5 LTS, which were connected directly with a 1000Mb/s ethernet connection.

For instructions on how to perform the experiments in our case studies, we refer to the detailed instructions in the project README.md file.

\section{Ethical Considerations}

This work presents methods for creating highly secure protocol implementations, which is for the benefit of the entire security community. We do not see dangers of potential misuse of AutoTam itself. A malicious actor would not benefit much from access to AutoTam or our description of it in this paper.

The code provided to replicate our performance benchmarks includes a simple UDP client to send messages, which could in theory be used to attempt to overload a network. We do not believe that our very simple UDP client is providing anything which would not already be easily implementable with a little bit of programming knowledge. We believe that providing the exact code used in order for our results to be fully replicable is more beneficial than hiding the implementation of our UDP client.

\section{Code Examples} \label{app:code_autotam}

We show our full AutoTam code for the DH example in Listing~\ref{lst:dh_initiator_autotam} and Listing~\ref{lst:dh_responder_autotam}. Finally, we show the generated Tamarin output for the DH example in Listing~\ref{lst:dh_tamarin_output}.

\section{Full Translation Table} \label{sec:appendix1_dsl}

Table~\ref{tab:dsl_translation} shows the full translation table.
The left-hand side shows available functions in the AutoTam language. The right-hand side shows the expression which is added to the rule when building the Tamarin rules for a transition. The type describes how each expression is added to the Tamarin rule. ACTION is used for equality and inequality restrictions which in Tamarin are encoded using the action facts Eq and Neq. SET describes an assignment, and will use the expression to the right of the arrow in place of the expression on the left. IN and OUT types are facts which are added to the precondition and postcondition of the rule, respectively. These include network operations which make terms visible to the network adversary through In and Out facts, fresh terms in the form of Fr, and private and public key provisioning.

\newpage

\begin{lstlisting}[numbers=none, numbersep=4pt, frame=lines, language={C}, basicstyle=\ttfamily\scriptsize, caption={DH Initiator code in AutoTam}, label={lst:dh_initiator_autotam}, commentstyle=\color{codegray}, captionpos=b]
START {
    Input {
        sock = Arg_Int()
        I = Arg(8)
        I_be = htobe(I)
        sk_i = Arg_priv_sign(32)
        R = Arg(8)
        R_be = htobe(R)
        pk_r = Arg_pub_sign(32)
        pk_i = Arg_pub_sign(32)
        x = GenSecret(32)
        gx = EC25519_PublicKey(x)
    }
}

Transition (START, State1)

State1 {
    Output {
        const1 = Const(1, 1)
        m1 = concat(const1, I_be)
        m1 = concat(m1, R_be)
        m1 = concat(m1, gx)
        send(sock, m1)
    }
    Input {
        m2 = recv(sock, 148)
    }
}

Transition (State1, DONE) {
    const10 = Const(1, 10)    
    msg_type, m2_tail_1 = split(m2, 1)
    RequireEq(msg_type, const10)
    sig_msg2, msg2 = split(m2_tail_1, 64)
    Ed25519_ValidateSig(pk_r, msg2, sig_msg2)
    const2 = Const(1, 2)
    msg_type_inner, msg2_tail_1 = split(msg2, 1)
    RequireEq(msg_type_inner, const2)
    responder_name, msg2_tail_2 = split(msg2_tail_1, 8)
    RequireEq(responder_name, R_be)
    initiator_name, msg2_tail_3 = split(msg2_tail_2, 8)
    RequireEq(initiator_name, I_be)
    pk_initiator, gy = split(msg2_tail_3, 32)
    RequireEq(pk_initiator, gx)
    RequireSize(gy, 32)
    gxy = EC25519_DH(x, gy)
}

DONE {
    Output {
        const3 = Const(1, 3)
        msg3 = concat(const3, I_be)
        msg3 = concat(msg3, R_be)
        msg3 = concat(msg3, gx)
        msg3 = concat(msg3, gy)
        sig_msg3 = Ed25519_SignMessage(sk_i, msg3)
        const10 = Const(1, 10)
        m3 = concat(const10, sig_msg3)
        m3 = concat(m3, msg3)
        send(sock, m3)
        Ret(gxy)
        Dummy(I)
        Dummy(R)
        Dummy(pk_i)
        Dummy(pk_r)
    }
}
\end{lstlisting}

\newpage

\begin{lstlisting}[numbers=none, numbersep=4pt, frame=lines, language={C}, basicstyle=\ttfamily\scriptsize, caption={DH Responder code in AutoTam}, label={lst:dh_responder_autotam}, commentstyle=\color{codegray}, captionpos=b]
START {
    Input {
        sock = Arg_Int()
        R = Arg(8)
        R_be = htobe(R)
        sk_r = Arg_priv_sign(32)
        I = Arg(8)
        I_be = htobe(I)
        pk_i = Arg_pub_sign(32)
        pk_r = Arg_pub_sign(32)
        m1 = recv(sock, 49)
    }
}

Transition (START, State1){
    const1 = Const(1, 1)
    
    msg_type, m1_tail_1 = split(m1, 1)
    RequireEq(msg_type, const1)
    initiator_name, m1_tail_2 = split(m1_tail_1, 8)
    RequireEq(initiator_name, I_be)
    responder_name, gx = split(m1_tail_2, 8)
    RequireEq(responder_name, R_be)
    RequireSize(gx, 32)
}

State1 {
    Output {
        y = GenSecret(32)
        gy = EC25519_PublicKey(y)
        gxy = EC25519_DH(y, gx)
        const2 = Const(1, 2)
        msg2 = concat(const2, R_be)
        msg2 = concat(msg2, I_be)
        msg2 = concat(msg2, gx)
        msg2 = concat(msg2, gy)
        sig_msg2 = Ed25519_SignMessage(sk_r, msg2)
        const10 = Const(1, 10)
        m2 = concat(const10, sig_msg2)
        m2 = concat(m2, msg2)
        send(sock, m2)
    }
    Input {
        m3 = recv(sock, 148)
    }
}

Transition (State1, DONE) {
    const10 = Const(1, 10)    
    msg_type, m3_tail_1 = split(m3, 1)
    RequireEq(msg_type, const10)
    sig_msg3, msg3 = split(m3_tail_1, 64)
    Ed25519_ValidateSig(pk_i, msg3, sig_msg3)
    const3 = Const(1, 3)
    msg_type_inner, msg3_tail_1 = split(msg3, 1)
    RequireEq(msg_type_inner, const3)
    initiator_name, msg3_tail_2 = split(msg3_tail_1, 8)
    RequireEq(initiator_name, I_be)
    responder_name, msg3_tail_3 = split(msg3_tail_2, 8)
    RequireEq(responder_name, R_be)
    pk_initiator, pk_responder = split(msg3_tail_3, 32)
    RequireEq(pk_initiator, gx)
    RequireEq(pk_responder, gy)
}

DONE {
    Output {
        Ret(gxy)
        Dummy(I)
        Dummy(R)
        Dummy(pk_i)
        Dummy(pk_r)
    }
}
\end{lstlisting}

\begin{figure*}
\begin{lstlisting}[numbers=none, numbersep=4pt, frame=lines, language={C}, basicstyle=\ttfamily\scriptsize, caption={Generated Tamarin model for DH example without Lemmas and with Actions removed except for equality restrictions. Line breaks have been added for readability. See code repository for full example.}, label={lst:dh_tamarin_output}, commentstyle=\color{codegray}, captionpos=b]
theory SignedDH
begin

builtins: signing, diffie-hellman


rule Init_Signing_LTK:
    let sign_pk = pk(~sign_ltk) in
    [ Fr(~sign_ltk) ]
  --[]->
    [ Out(sign_pk), !SignLtk(~sign_ltk), !SignPk(sign_pk) ]

rule Leak_Signing_LTK:
    [ !SignLtk(~sign_ltk) ]
  --[ Leak(pk(~sign_ltk)) ]->
    [ Out(~sign_ltk) ]

rule Create_arg_priv:
    let dh_pk = 'g'^~dh_ltk in
    [ Fr(~dh_ltk) ]
  --[]->
    [ Out(dh_pk), !DHLtk(~dh_ltk), !DHPk(dh_pk) ]

rule Leak_DH_LTK:
    [ !DHLtk(~dh_ltk) ]
  --[ Leak('g'^~dh_ltk) ]->
    [ Out(~dh_ltk) ]

rule Initiator_START_to_State1:
    let gx = 'g'^~x
    in 
    [ In(sock), In(I), !SignLtk(~sk_i), In(R), !SignPk(pk_r), !SignPk(pk_i), Fr(~x) ]
  --[ ]->
    [ Out(<'1', I, R, gx>), InitiatorState1(sock, I, I, ~sk_i, R, R, pk_r, pk_i, ~x, gx) ]

rule Initiator_State1_to_DONE:
    let msg3222 = <'3', I_be, R_be, gx, gy>
        sig_msg3 = sign(msg3222, ~sk_i)
        m32 = <'10', sig_msg3, '3', I_be, R_be, gx, gy>
    in 
    [ InitiatorState1(sock, I, I_be, ~sk_i, R, R_be, pk_r, pk_i, ~x, gx), 
        In(<msg_type, sig_msg2, msg_type_inner, responder_name, initiator_name, pk_initiator, gy>) ]
  --[ Eq(msg_type, '10'), 
        Eq(verify(sig_msg2, <msg_type_inner, responder_name, initiator_name, pk_initiator, gy>, pk_r), true), 
        Eq(msg_type_inner, '2'), Eq(responder_name, R_be), Eq(initiator_name, I_be), Eq(pk_initiator, gx) ]->
    [ Out(m32) ]

rule Responder_START_to_State1:
    let gy = 'g'^~y
        gxy = gx^~y
    in 
    [ In(sock), In(R), !SignLtk(~sk_r), In(I), !SignPk(pk_i), !SignPk(pk_r), 
        In(<msg_type, initiator_name, responder_name, gx>), Fr(~y) ]
  --[ Eq(msg_type, '1'), Eq(initiator_name, I), Eq(responder_name, R) ]->
    [ Out(<'10', sign(<'2', R, I, gx, gy>, ~sk_r), '2', R, I, gx, gy>), 
        ResponderState1(R, R, I, I, pk_i, pk_r, gy, gxy, gx) ]

rule Responder_State1_to_DONE:
    [ ResponderState1(R, R_be, I, I_be, pk_i, pk_r, gy, gxy, gx), 
        In(<msg_type, sig_msg3, msg_type_inner, initiator_name, responder_name, pk_initiator, pk_responder>) ]
  --[ Eq(msg_type, '10'), 
        Eq(verify(sig_msg3, <msg_type_inner, initiator_name, responder_name, pk_initiator, pk_responder>, pk_i), true), 
        Eq(msg_type_inner, '3'), Eq(initiator_name, I_be), Eq(responder_name, R_be), Eq(pk_initiator, gx), 
        Eq(pk_responder, gy) ]->
    [  ]

end

\end{lstlisting}
\end{figure*}

\begin{table*}
    \centering
    \caption{Full Translation Table}
    \begin{tabular}{l c l}
        \toprule
        \textbf{Function call (AutoTam language)} & \multicolumn{2}{c}{\textbf{Translation (Tamarin)}} \\
        \cmidrule(r){2-3}
        & \textbf{Type} & \textbf{Expression} \\
        \midrule
        p = Chacha20\_Poly1305\_dec(k, n, c, d, m) & [ACTION] & Eq(aead\_val(k, n, c, d, m), true) \\
        & [SET] & p $\rightarrow$ aead\_dec(k, n, c) \\
        \midrule
        c, m = Chacha20\_Poly1305\_enc(k, n, p, d) & [SET] & c $\rightarrow$ aead\_enc(k, n, p, d), \\
        & [SET] & m $\rightarrow$ aead\_mac(k, n, p, d) \\
        gx = EC25519\_PublicKey(x) & [SET] & gx $\rightarrow$ 'g'\texttt{\^}x \\
        gxy = EC25519\_DH(x, gy) & [SET] & gxy $\rightarrow$ gy\texttt{\^}x \\
        sig = Ed25519\_SignMessage(sk, m) & [SET] & sig $\rightarrow$ sign(m, sk) \\
        Ed25519\_ValidateSig(pk, m, sig) & [ACTION] & Eq(verify(sig, m, pk), true) \\
        a = Blake2s\_keyhash(k, m, SIZE) & [SET] & a $\rightarrow$  keyhash(k, m) \\
        a = Blake2s\_hmac(k, m) & [SET] & a $\rightarrow$ hmac(k, m) \\
        x = GenSecret(SIZE) & [IN] & Fr($\sim$x) \\
		m = recv(fd, SIZE) & [IN] & In(m) \\
        m, addr = recvfrom(fd, SIZE) & [IN] & In(m), In(addr) \\
		send(fd, m) & [OUT] & Out(m) \\
        sendto(fd, addr, m) & [OUT] & Out(addr), Out(m) \\
        a = concat(b, c) & [SET] & a $\rightarrow$ $\langle$b, c$\rangle$ \\
        b, c = split(a, SIZE) & [SET] & a $\rightarrow$ $\langle$b, c$\rangle$ \\
        t = timestamp() & [SET] & t $\rightarrow$ \$t \\
        s = String('example') & [SET] & s $\rightarrow$ 'example' \\
        n = Constant(NUM, SIZE) & [SET] & n $\rightarrow$ 'NUM' \\
        RequireEq(a, b) & [ACTION] & Eq(a, b) \\
        RequireNEq(a, b) & [ACTION] & Neq(a, b) \\
        RequireSize(x, SIZE) & - & - \\
        a = size(b) & [SET] & a $\rightarrow$ \$a \\
        a = Pad16(b) & [SET] & a $\rightarrow$ b \\
        a = htobe(b) & [SET] & a $\rightarrow$ b \\
        n = Counter64() & [SET] & n $\rightarrow$ \$n \\
        m = Increment(n) & [SET] & m $\rightarrow$ \$m \\
        x = Arg(SIZE) & [IN] & In(x) \\
        x = Arg\_Int() & [IN] & In(x) \\
        x = Arg\_priv\_dh(SIZE) & [IN] & !DHLtk($\sim$x) \\
        x = Arg\_pub\_dh(SIZE) & [IN] & !DHPk(x) \\
        x = Arg\_priv\_sign(SIZE) & [IN] & !SignLtk($\sim$x) \\
        x = Arg\_pub\_sign(SIZE) & [IN] & !SignPk(x) \\
        Ret(x) & - & - \\
        Dummy(x) & - & - \\
        \bottomrule
    \end{tabular}
    \label{tab:dsl_translation}
\end{table*}

\begin{figure*}
\begin{lstlisting}[numbers=none, numbersep=4pt, frame=lines, language={C}, basicstyle=\ttfamily\small, caption={AutoTam Syntax in Backus-Naur form}, label={lst:syntax_backusnaur}, commentstyle=\color{codegray}, captionpos=b]
<Program> ::= <State> <Program> 
            | <Transition> <Program> 
            | ""

<State> ::= <StateName> { Output { <StatementList> } Input { <StatementList> } }
          | <StateName> { Output { <StatementList> } }
          | <StateName> { Input { <StatementList> } }
          | <StateName> {}
		  
<Transition> ::= Transition (<StateName>, <StateName>) { <StatementList> }
               | Transition (<StateName>, <StateName>)
			   
<StatementList> ::= <Statement> \n <StatementList>
                  | ""
				  
<Statement> ::= <VariableList> = <FunctionName>(<VariableList>)
              | <VariableList> = <FunctionName>()
                               | <FunctionName>(<VariableList>)
			  
<VariableList> ::= <VariableName>, <VariableList>
                 | <VariableName>
\end{lstlisting}
\end{figure*}

\section{Proof of Lemma 2} \label{sec:proof_lemma2}

\begin{proof}[Proof Lemma 2.]
Assume the lemma is false. Then there must exist some scenario $C$ with a transition truncated global execution trace $w_e^*$, without there being a complete global execution trace $w_e^\# \in \tau_e(C)$ so that $\Psi(w_e) = \Psi(w_e^*)$.
Then either an input state segment was partially completed in $w_e^*$, or an input state segment was completed by a participant $p$ without any following transition being performed by $p$ in $w_e^*$, or an output segment wasn't fully completed after a completed transition by some participant $p$ in $w_e*$.

If an input state segment was partially completed in $w_e^*$, or an input state segment was completed without any following transition in $w_e^*$, then there must exist a trace $w_e^\# \in \tau_e(C)$ where that input state segment was never performed. Executed input state segments do not affect the execution of other participants since no additional information becomes visible to the network. Additionally, executed input state segments are not visible on the completed transition trace. Therefore, it holds that $\Psi(w_e) = \Psi(w_e^*)$.

If an output segment wasn't fully completed after a completed transition by some participant $p$ in $w_e^*$, then requirement~\ref{req:no_stop_out} rules out all of the statements which impose restrictions in the symbolic model\footnote{In the symbolic model, we exclude network errors, leaving only equality tests (including ones imposed by signature verification and AEAD decryption) as the statements which may ''halt'' execution.}, leaving no restrictions to prevent further execution. Therefore, the same scenario must permit a trace $w_e^\# \in \tau_e(C)$ where execution continues, completing the output state segment. Any information revealed to the network in $w_e^*$ is also revealed in $w_e^\#$, thus further execution which is possible for other participants in $w_e^*$ remains possible in $w_e^\#$. Additionally, executed output state segments are not visible on the completed transition trace. Therefore, it holds that $\Psi(w_e) = \Psi(w_e^*)$.
\end{proof}

\section{Formal Syntax Definition}

A more formal definition of the AutoTam syntax is given in Listing~\ref{lst:syntax_backusnaur} in Backus-Naur form. Alphanumeric strings can be used in place of symbols ending in Name. Only a \texttt{FunctionName} which appears in the left-hand side of Table~\ref{tab:dsl_translation}, with the correct number of parameters and output variables, will be considered valid. In addition, the function translation may not be of type \texttt{IN} if placed in the Output part of a State, of type \texttt{OUT} if placed in the Input part of a State, or of either type \texttt{IN} or \texttt{OUT} when placed within a Transition. Note that transitions themselves do not have names, while states and variables do.

\end{document}